\documentclass[oneside,english]{amsart}
\usepackage[T1]{fontenc}
\usepackage[latin9]{inputenc}
\usepackage{geometry}
\usepackage{amstext}
\usepackage{amsthm}
\usepackage{amssymb}
\usepackage{graphicx}
\usepackage{setspace}
\makeatletter
\numberwithin{equation}{section}
\numberwithin{figure}{section}
\theoremstyle{plain}
\newtheorem{thm}{\protect\theoremname}
\theoremstyle{plain}
\newtheorem{prop}[thm]{\protect\propositionname}

\usepackage{hyperref}
\newcommand{\dd}{\mathrm{d}}

\makeatother

\usepackage{babel}
\providecommand{\propositionname}{Proposition}
\providecommand{\theoremname}{Theorem}

\begin{document}
\title{Generalised model of wear in contact problems: the case of oscillatory
load}
\author{Dmitry Ponomarev$^{1,2,3}$}
\begin{abstract}
In this short paper, we consider a sliding punch problem under recently
proposed model of wear which is based on the Riemann-Liouville fractional
integral relation between pressure and worn volume, and incorporates
another additional effect pertinent to relaxation. A particular case
of oscillatory (time-harmonic) load is studied. The time-dependent
stationary state is identified in terms of eigenfunctions of an auxiliary
integral operator. Convergence to this stationary state is quantified.
Moreover, numerical simulations have been conducted in order to illustrate
the obtained results and study qualitative dependence on two main
model parameters.
\end{abstract}

\maketitle

\section{Introduction\label{sec:intro}}

\footnotetext[1]{FACTAS team, Centre Inria d'Universit{\'e} C{\^o}te d'Azur, France}\footnotetext[2]{St. Petersburg Department of Steklov Mathematical Institute of Russian Academy of Sciences, Russia}\footnotetext[3]{Contact: dmitry.ponomarev@inria.fr}Due
to its practical importance, contact problems with wear has been an
area of active research for decades. A problem when indented wearable
punch slides with a constant speed on an elastic layer or a half-space
is a classical setting (see e.g. \cite{AlblKuip1,AlblKuip2}). Numerous
empirical laws of wear were proposed to fit predictions of different
models \cite{ArgChai1,ArgChai2,ArgFad,Fepp,Goryach,Kov,Zhu} to experiments
such as a pin on a rotating disk. Recently, a Riemann-Liouville relation,
generalising the classical simple integral relation between worn volume
and pressure, was motivated in \cite{Argat}. Mathematical analysis
of that model and its further extension was given in \cite{Pon}.
Namely, in addition to the fractional order integration \cite{GorMain},
another generalisation of the model has been incorporated aiming to
account for possible relaxation effects \cite{YevtPyr}. The long-time
behavior of the pressure profile has been investigated in the set-ups
where the exterior load was either constant or eventually constant
(i.e. the so-called transitional load describing a smooth switch between
two values over some finite time interval).

In the present work, we are concerned with analysis of the long-time
behavior of the solution of the generalised model in the situation
where the exterior load is time-harmonic. The applied analysis also
automatically applies to the classical model (with worn volume and
pressure related through the basic Archard's law \cite{Arch,Gal})
as a particular case.

The structure of the paper is as follows. In Section \ref{sec:model},
we briefly recall the previously proposed model as well as some auxiliary
functions and their properties that are going to be essential for
the present work. Section \ref{sec:analysis} is dedicated to the
analysis of the solution of the model: we will identify time-dependent
stationary state for the pressure profile and estimate the convergence
rate depending on the value of the model parameter $\alpha$. We illustrate
the obtained results numerically in Section \ref{sec:numerics} and
conclude with their discussion in Section \ref{sec:conclus}.

\section{Model\label{sec:model}}

According to \cite{AleksKov1,ArgChai1,ArgChai2,ArgFad,Gal,Kom,Vor},
the pressure under the punch satisfies the following equation for
displacements 
\begin{equation}
\eta p\left(x,t\right)+\int_{-a}^{a}K\left(x-\xi\right)p\left(\xi,t\right)\dd\xi=\delta\left(t\right)-w\left[p\right]\left(x,t\right)-\Delta\left(x\right),\hspace{1em}x\in\left(-a,a\right),\hspace{1em}t\geq0,\label{eq:displ_bal}
\end{equation}
and the force equilibrium condition
\begin{equation}
\int_{-a}^{a}p\left(x,t\right)\dd x=P\left(t\right),\hspace{1em}t\geq0.\label{eq:p_equil}
\end{equation}

The contact load is denoted $P\left(t\right)$ and, in the present
work, we study its particular form, namely, 
\begin{equation}
P\left(t\right)=P_{0}-P_{\Delta}+P_{\Delta}\cos\left(\omega t\right)\label{eq:P_t}
\end{equation}
for some given constants $P_{0}$, $P_{\Delta}$, $\omega>0$.

The interval $\left(-a,a\right)$ corresponds to the contact area
under the punch, $K\left(x\right)$ is a kernel function of the ``pressure-to-displacement''
operator. Such an operator stems from the Green's function pertinent
to a given geometry. In particular, we are interested in an elastic
half-space problem which is the limiting case of the thick-layer problem.
In this case, we have 
\begin{equation}
K\left(x\right):=-\log\left|x\right|+C_{K},\label{eq:K_def}
\end{equation}
for some constant $C_{K}>\log a$.

The illustration of the problem geometry is given in Figure \ref{fig:geom}.

\begin{figure}
\includegraphics[scale=0.5]{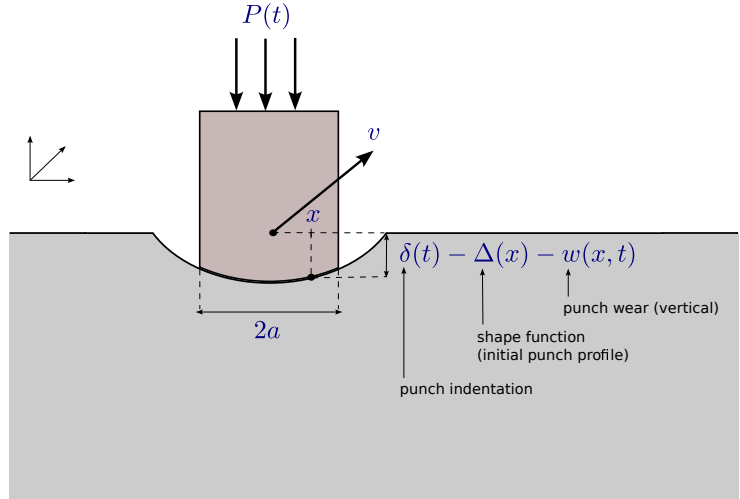}

\caption{\label{fig:geom} Illustration of the geometry of the problem}
\end{figure}

The first term on the left-hand side of (\ref{eq:displ_bal}) accounts
for the additional deformation due to the presence of a coating or
to model surface roughness \cite{AleksKov2,Kov}. The strength of
this effect is measured by the constant $\eta>0$.\\
The initial punch profile $\Delta\left(x\right)$ is a known function
whereas the punch indentation $\delta\left(t\right)$ is a function
of only time. Its initial value $\delta\left(0\right)$ can be found
from solving
\begin{equation}
\eta p\left(x,0\right)+\int_{-a}^{a}K\left(x-\xi\right)p\left(\xi,0\right)\dd\xi=\delta\left(0\right)-\Delta\left(x\right),\hspace{1em}x\in\left(-a,a\right),\label{eq:p0_int_eq}
\end{equation}
and requiring that $\int_{-a}^{a}p\left(x,0\right)\dd x=P_{0}$, as
follows from (\ref{eq:displ_bal}) and (\ref{eq:p_equil}), respectively,
evaluated at $t=0$.\\
Finally, $w\left[p\right]\left(x,t\right)$ is the wear term which
is an operator acting on the contact pressure $p\left(x,t\right)$.
Following the discussion in \cite[Sec. 3]{Pon}, we take it as 
\begin{equation}
w\left[p\right]\left(x,t\right)=-\nu\mu^{1/\alpha-1}\int_{0}^{t}\mathcal{E}_{\alpha}\left(\mu^{1/\alpha}\left(t-\tau\right)\right)p\left(x,\tau\right)\dd\tau,\label{eq:w_p_term}
\end{equation}
where $\mu>0$ is a constant and the special function $\mathcal{E}_{\alpha}$
can be defined as
\begin{equation}
\mathcal{E_{\alpha}}\left(x\right):=\frac{\alpha}{x}\sum_{k=1}^{\infty}\frac{\left(-1\right)^{k}kx^{\alpha k}}{\Gamma\left(\alpha k+1\right)},\hspace{1em}x>0,\hspace{1em}\alpha>0,\label{eq:cal_E_a_def}
\end{equation}
with $\Gamma$ being the Gamma function. Note that, in particular
case, when $\alpha=1,$we have $\mathcal{E}_{1}\left(x\right)=-\exp\left(-x\right)$.\\
We will make use of the following asymptotics
\begin{equation}
\mathcal{E_{\alpha}}\left(x\right)=-\frac{\alpha}{\Gamma\left(1+\alpha\right)}\frac{1}{x^{1-\alpha}}+\mathcal{O}\left(\frac{1}{x^{1-2\alpha}}\right),\hspace{1em}\hspace{1em}\left|x\right|\ll1,\label{eq:cal_E_a_small}
\end{equation}
\begin{equation}
\mathcal{E_{\alpha}}\left(x\right)=\begin{cases}
-\frac{\alpha}{\Gamma\left(1-\alpha\right)}\frac{1}{x^{\alpha+1}}+\mathcal{O}\left(\frac{1}{x^{2\alpha+1}}\right), & \alpha\in\left(0,1\right)\cup\left(1,2\right),\\
-\exp\left(-x\right), & \alpha=1,
\end{cases}\hspace{1em}\hspace{1em}x\gg1,\label{eq:cal_E_a_large}
\end{equation}
as well as the integral relation
\begin{equation}
\int_{0}^{x_{0}}\mathcal{E}_{\alpha}\left(\lambda^{1/\alpha}x\right)\dd x=\lambda^{-1/\alpha}\left[E_{\alpha}\left(-\lambda x_{0}^{\alpha}\right)-1\right],\hspace{1em}x_{0}>0,\hspace{1em}\alpha>0,\hspace{1em}\lambda>0,\label{eq:cal_E_a_int}
\end{equation}
where $E_{\alpha}$ is the Mittag-Leffler function defined as
\begin{equation}
E_{\alpha}\left(z\right):=\sum_{k=0}^{\infty}\frac{z^{k}}{\Gamma\left(\alpha k+1\right)},\hspace{1em}z\in\mathbb{C},\hspace{1em}\alpha>0.\label{eq:E_a_def}
\end{equation}
In particular, we have $E_{1}\left(z\right)=\exp z$, $z\in\mathbb{C}$,
and $E_{\alpha}\left(0\right)=1$, $\alpha>0$. Moreover, the following
useful asymptotic holds true
\begin{equation}
E_{\alpha}\left(-x\right)=\begin{cases}
\frac{1}{\Gamma\left(1-\alpha\right)}\frac{1}{x}+\mathcal{O}\left(\frac{1}{x^{2}}\right), & \alpha\in\left(0,1\right)\cup\left(1,2\right),\\
\exp\left(-x\right), & \alpha=1,
\end{cases}\hspace{1em}\hspace{1em}x\gg1.\label{eq:E_a_large}
\end{equation}
References for the above mentioned results can be found in \cite[Appendix A]{Pon}.

Finally, we recall from \cite[Sec. 3]{Pon} that when $\alpha=1$
and $\mu=0$ (or more precisely, in the limit of $\mu\searrow0$),
the relation (\ref{eq:w_p_term}) reduces to what is consistent with
the classical Archard's law \cite{Arch}: 
\begin{equation}
w\left[p\right]\left(x,t\right)=-\nu\int_{0}^{t}p\left(x,\tau\right)\dd\tau.\label{eq:Arch_law}
\end{equation}

\section{Analysis\label{sec:analysis}}

The model (\ref{eq:displ_bal})--(\ref{eq:p_equil}) has been rigorously
analysed in \cite[Sec. 4]{Pon}. We adapt here a general theory \cite[Thm 6]{Pon}
using the result of \cite[Prop. 14]{Pon} valid for the particular
form of the kernel function (\ref{eq:K_def}). Namely, we have the
following theorem.
\begin{thm}
\label{thm:main} Assume that $\mu\geq0$, $\eta$, $\nu>0$, $\alpha\in\left(0,2\right)$,
$a\neq2$, and $p\left(\cdot,0\right)\in L^{2}\left(-a,a\right)$
solves (\ref{eq:p0_int_eq}) with $K$ given by (\ref{eq:K_def})
and $\delta\left(0\right)$ such that $\int_{-a}^{a}p\left(x,0\right)\dd x=P_{0}$
for some $P_{0}>0$. Suppose $P\left(t\right)$ is as in (\ref{eq:P_t})
and $w\left[p\right]$ is defined in (\ref{eq:w_p_term}). Then, the
unique solution $p\in C_{b}\left(\mathbb{R}_{+};L^{2}\left(-a,a\right)\right)$
of (\ref{eq:displ_bal}) satisfying (\ref{eq:p_equil}) is given by
\begin{equation}
p\left(x,t\right)=\frac{P\left(t\right)}{2a}+\sum_{k=1}^{\infty}d_{k}\left(t\right)\phi_{k}\left(x\right),\hspace{1em}x\in\left(-a,a\right),\hspace{1em}t\geq0,\label{eq:p_sol}
\end{equation}
where
\begin{align}
d_{k}\left(t\right):= & d_{k}^{0}\left[1+\frac{\nu}{\mu\left(\eta+\sigma_{k}\right)+\nu}\left(E_{\alpha}\left(-\left(\mu+\frac{\nu}{\eta+\sigma_{k}}\right)t^{\alpha}\right)-1\right)\right]-\frac{l_{k}}{2a\left(\eta+\sigma_{k}\right)}\left(P\left(t\right)-P\left(0\right)\right)\label{eq:d_k_def}\\
 & -\frac{\nu l_{k}}{2a\left(\eta+\sigma_{k}\right)^{2}}\left(\mu+\frac{\nu}{\eta+\sigma_{k}}\right)^{1/\alpha-1}\int_{0}^{t}\mathcal{E}_{\alpha}\left(\left(\mu+\frac{\nu}{\eta+\sigma_{k}}\right)^{1/\alpha}\left(t-\tau\right)\right)\left[P\left(\tau\right)-P\left(0\right)\right]\dd\tau,\hspace{1em}k\geq1,\nonumber 
\end{align}
\begin{equation}
d_{k}^{0}:=\int_{-a}^{a}p\left(\xi,0\right)\phi_{k}\left(\xi\right)\dd\xi,\hspace{1em}\hspace{1em}l_{k}:=\frac{1}{2a}\int_{-a}^{a}\int_{-a}^{a}K\left(\zeta-\xi\right)\phi_{k}\left(\xi\right)\dd\zeta\dd\xi,\hspace{1em}k\geq1.\label{eq:d_k_0_l_k_def}
\end{equation}
Here, $\mathcal{E}_{\alpha}$ and $E_{\alpha}$ are as in (\ref{eq:cal_E_a_def})
and (\ref{eq:E_a_def}), respectively, whereas $\phi_{k}\in L_{0}^{2}\left(-a,a\right)$,
$\sigma_{k}>0$, $k\geq1$, are normalised eigenfunctions and eigenvalues
of the compact self-adjoint operator 
\begin{equation}
\mathcal{K}_{2}\left[\phi\right]\left(x\right):=\int_{-a}^{a}\left[K\left(x-\xi\right)-\frac{1}{2a}\int_{-a}^{a}\left(K\left(\zeta-x\right)+K\left(\zeta-\xi\right)\right)\dd\zeta\right]\phi\left(\xi\right)\dd\xi\label{eq:K2_def}
\end{equation}
defined on the functional space
\begin{equation}
L_{0}^{2}\left(-a,a\right):=\left\{ f\in L^{2}\left(-a,a\right):\;\int_{-a}^{a}f\left(x\right)\dd x=0\right\} .\label{eq:L0_def}
\end{equation}
\end{thm}

\begin{proof}
The statement of the theorem is merely a rephrasement of several results
from \cite{Pon}. First of all, it is straightforward (see also \cite[Prop. 11]{Pon})
to verify that the kernel function (\ref{eq:K_def}) and the exterior
load (\ref{eq:P_t}) satisfy all assumptions of \cite[Thm 6]{Pon}.
Then, thanks to \cite[Prop. 14]{Pon}, we observe that a form of the
solution given by \cite[Thm 6]{Pon} simplifies since $\text{Ker }\mathcal{K}_{2}$,
the kernel space of the auxiliary operator $\mathcal{K}_{2}$, is
empty. This last part calls for the additional assumption $a\neq2$
appearing in the formulation. Finally, the fact that $\sigma_{k}>0$
for $k\geq1$ follows from \cite[Prop. 4]{Pon} and another use of
\cite[Prop. 14]{Pon}.
\end{proof}
We now proceed with the main goal of the paper. We identify the stationary
state and perform long-time behaviour analysis to study qualitative
character and speed of the convergence of the solution to this stationary
state.
\begin{prop}
\label{prop:stat_stat_conv}Under assumptions of Theorem (\ref{thm:main}),
there exists $\delta_{p}\in C\left(\mathbb{R}_{+};L^{2}\left(-a,a\right)\right)$,
$\left\Vert \delta_{p}\left(\cdot,t\right)\right\Vert _{L^{2}\left(-a,a\right)}\longrightarrow0$
as $t\rightarrow+\infty$, such that the solution $p\left(x,t\right)$
to the model (\ref{eq:displ_bal})--(\ref{eq:p_equil}) can be written
as
\begin{equation}
p\left(x,t\right)=\widetilde{p}_{\infty}\left(x\right)-W_{0}\left(x\right)\cos\left(\omega t-\psi\left(x\right)\right)+\delta_{p}\left(x,t\right)=:p_{\infty}\left(x,t\right)+\delta_{p}\left(x,t\right),\label{eq:p_long_t}
\end{equation}
where
\begin{equation}
\widetilde{p}_{\infty}\left(x\right):=\frac{1}{2a}\left(P_{0}-P_{\Delta}\right)+\sum_{k=1}^{\infty}\left(d_{k}^{0}\frac{\mu\left(\eta+\sigma_{k}\right)}{\mu\left(\eta+\sigma_{k}\right)+\nu}+\frac{P_{\Delta}\mu l_{k}}{2a\left[\mu\left(\eta+\sigma_{k}\right)+\nu\right]}\right)\phi_{k}\left(x\right),\label{eq:p_inf_def}
\end{equation}
\begin{equation}
W_{0}\left(x\right):=\left[W_{1}\left(x\right)+W_{2}\left(x\right)\right]^{1/2},\hspace{1em}\hspace{1em}\psi\left(x\right):=\textnormal{sign}\left(W_{2}\left(x\right)\right)\arccos\frac{W_{1}\left(x\right)}{W_{0}\left(x\right)},\label{eq:W0_psi_def}
\end{equation}
\begin{equation}
W_{1}\left(x\right):=-1+\frac{P_{\Delta}}{2a}\sum_{k=1}^{\infty}\frac{l_{k}}{\eta+\sigma_{k}}\left[\left(\mu+\frac{\nu}{\eta+\sigma_{k}}\right)^{1/\alpha-1}\frac{\nu}{\eta+\sigma_{k}}C_{k}^{c,\omega}+1\right]\phi_{k}\left(x\right),\label{eq:W1_def}
\end{equation}
\begin{equation}
W_{2}\left(x\right):=\frac{P_{\Delta}}{2a}\sum_{k=1}^{\infty}\frac{\nu l_{k}}{\left(\eta+\sigma_{k}\right)^{2}}\left(\mu+\frac{\nu}{\eta+\sigma_{k}}\right)^{1/\alpha-1}C_{k}^{s,\omega}\phi_{k}\left(x\right),\label{eq:W2_def}
\end{equation}
\begin{equation}
C_{k}^{c,\omega}:=\int_{0}^{\infty}\mathcal{E}_{\alpha}\left(\left(\mu+\frac{\nu}{\eta+\sigma_{k}}\right)^{1/\alpha}\tau\right)\cos\left(\omega\tau\right)\dd\tau,\hspace{1em}\hspace{1em}C_{k}^{s,\omega}:=\int_{0}^{\infty}\mathcal{E}_{\alpha}\left(\left(\mu+\frac{\nu}{\eta+\sigma_{k}}\right)^{1/\alpha}\tau\right)\sin\left(\omega\tau\right)\dd\tau.\label{eq:Cc_Cs_def}
\end{equation}
Moreover, we have
\begin{equation}
\left\Vert \delta_{p}\left(\cdot,t\right)\right\Vert _{L^{2}\left(-a,a\right)}=\mathcal{O}\left(\exp\left(-\left(\mu+\frac{\nu}{\eta+\sigma_{1}}\right)t\right)\right),\hspace{1em}t\gg1,\hspace{1em}\alpha=1,\label{eq:dp_conv_alph_1}
\end{equation}

\begin{equation}
\left\Vert \delta_{p}\left(\cdot,t\right)\right\Vert _{L^{2}\left(-a,a\right)}=\mathcal{O}\left(\frac{1}{t^{\alpha}}\right),\hspace{1em}t\gg1,\hspace{1em}\alpha\in\left(0,1\right)\cup\left(1,2\right).\label{eq:dp_conv}
\end{equation}
\end{prop}

\begin{proof}
Plugging (\ref{eq:P_t}) into (\ref{eq:d_k_def}) and employing (\ref{eq:cal_E_a_int}),
we obtain
\begin{align*}
d_{k}\left(t\right)= & d_{k}^{0}-\frac{P_{\Delta}l_{k}}{2a\left(\eta+\sigma_{k}\right)}\left[\cos\left(\omega t\right)-1\right]+\frac{\nu}{\mu\left(\eta+\sigma_{k}\right)+\nu}\left(d_{k}^{0}+\frac{P_{\Delta}l_{k}}{2a\left(\eta+\sigma_{k}\right)}\right)\left[E_{\alpha}\left(-\left(\mu+\frac{\nu}{\eta+\sigma_{k}}\right)t^{\alpha}\right)-1\right]\\
 & -\frac{\nu l_{k}P_{\Delta}}{2a\left(\eta+\sigma_{k}\right)^{2}}\left(\mu+\frac{\nu}{\eta+\sigma_{k}}\right)^{1/\alpha-1}\left[\cos\left(\omega t\right)\int_{0}^{t}\mathcal{E}_{\alpha}\left(\left(\mu+\frac{\nu}{\eta+\sigma_{k}}\right)^{1/\alpha}\tau\right)\cos\left(\omega\tau\right)\dd\tau\right.\\
 & \left.-\sin\left(\omega t\right)\int_{0}^{t}\mathcal{E}_{\alpha}\left(\left(\mu+\frac{\nu}{\eta+\sigma_{k}}\right)^{1/\alpha}\tau\right)\sin\left(\omega\tau\right)\dd\tau\right].
\end{align*}
Note that, due to (\ref{eq:cal_E_a_large}), the integrals here are
converging even when $t\rightarrow+\infty$. Hence, using the definitions
in (\ref{eq:Cc_Cs_def}), we can write
\[
\int_{0}^{t}\mathcal{E}_{\alpha}\left(\left(\mu+\frac{\nu}{\eta+\sigma_{k}}\right)^{1/\alpha}\tau\right)\cos\left(\omega\tau\right)\dd\tau=C_{k}^{c,\omega}-\int_{t}^{\infty}\mathcal{E}_{\alpha}\left(\left(\mu+\frac{\nu}{\eta+\sigma_{k}}\right)^{1/\alpha}\tau\right)\cos\left(\omega\tau\right)\dd\tau,
\]
\[
\int_{0}^{t}\mathcal{E}_{\alpha}\left(\left(\mu+\frac{\nu}{\eta+\sigma_{k}}\right)^{1/\alpha}\tau\right)\sin\left(\omega\tau\right)\dd\tau=C_{k}^{s,\omega}-\int_{t}^{\infty}\mathcal{E}_{\alpha}\left(\left(\mu+\frac{\nu}{\eta+\sigma_{k}}\right)^{1/\alpha}\tau\right)\sin\left(\omega\tau\right)\dd\tau.
\]
Consequently, we arrive at
\begin{align}
d_{k}\left(t\right)= & d_{k}^{0}\frac{\mu\left(\eta+\sigma_{k}\right)}{\mu\left(\eta+\sigma_{k}\right)+\nu}+\frac{l_{k}P_{\Delta}}{2a}\frac{\mu\left(\eta+\sigma_{k}\right)}{\mu\left(\eta+\sigma_{k}\right)+\nu}-\frac{l_{k}P_{\Delta}}{2a}\frac{\nu}{\left(\eta+\sigma_{k}\right)^{2}}\left(\mu+\frac{\nu}{\eta+\sigma_{k}}\right)^{1/\alpha-1}\left(C_{k}^{c,\omega}+1\right)\cos\left(\omega t\right)\label{eq:d_k_prelim}\\
 & -\frac{l_{k}P_{\Delta}}{2a}\frac{\nu}{\left(\eta+\sigma_{k}\right)^{2}}\left(\mu+\frac{\nu}{\eta+\sigma_{k}}\right)^{1/\alpha-1}C_{k}^{s,\omega}\sin\left(\omega t\right)+r_{k}\left(t\right),\nonumber 
\end{align}
where 
\begin{align}
r_{k}\left(t\right):= & \frac{\nu}{\mu\left(\eta+\sigma_{k}\right)+\nu}\left(d_{k}^{0}+\frac{P_{\Delta}l_{k}}{2a\left(\eta+\sigma_{k}\right)}\right)E_{\alpha}\left(-\left(\mu+\frac{\nu}{\eta+\sigma_{k}}\right)t^{\alpha}\right)\label{eq:r_k_def}\\
 & +\frac{l_{k}P_{\Delta}}{2a}\frac{\nu}{\left(\eta+\sigma_{k}\right)^{2}}\left(\mu+\frac{\nu}{\eta+\sigma_{k}}\right)^{1/\alpha-1}\int_{t}^{\infty}\mathcal{E}_{\alpha}\left(\left(\mu+\frac{\nu}{\eta+\sigma_{k}}\right)^{1/\alpha}\tau\right)\cos\left(\omega\left(t-\tau\right)\right)\dd\tau.\nonumber 
\end{align}
Substitution of (\ref{eq:d_k_prelim}) into (\ref{eq:p_sol}) and
taking into account (\ref{eq:p_inf_def})--(\ref{eq:W2_def}) yields
(\ref{eq:p_long_t}) with 
\begin{equation}
\delta_{p}\left(x,t\right):=\sum_{k=1}^{\infty}r_{k}\left(t\right)\phi_{k}\left(x\right).\label{eq:dp_def}
\end{equation}

It now remains to deduce (\ref{eq:dp_conv_alph_1})--(\ref{eq:dp_conv}).
To this effect, we first note that, due to the mutual orthogonality
of functions $\left\{ \phi_{k}\right\} _{k=1}^{\infty}$, we have
\begin{equation}
\left\Vert \delta_{p}\left(\cdot,t\right)\right\Vert _{L^{2}\left(-a,a\right)}=\left(\sum_{k=1}^{\infty}\left|r_{k}\left(t\right)\right|^{2}\right)^{1/2}.\label{eq:dp_Parseval}
\end{equation}

Using the fact that $0<\sigma_{k}\leq\sigma_{k-1}$ for $k>1$, we
can estimate, for sufficiently large $t>0$,
\begin{align}
\left(\sum_{k=1}^{\infty}\left|\frac{\nu}{\mu\left(\eta+\sigma_{k}\right)+\nu}E_{\alpha}\left(-\left(\mu+\frac{\nu}{\eta+\sigma_{k}}\right)t^{\alpha}\right)d_{k}^{0}\right|^{2}\right)^{1/2}\label{eq:dp_tmp_estim1}\\
\leq\frac{\nu}{\mu\eta+\nu}\left|E_{\alpha}\left(-\left(\mu+\frac{\nu}{\eta+\sigma_{1}}\right)t^{\alpha}\right)\right|\left\Vert p_{0}\left(\cdot,t\right)\right\Vert _{L^{2}\left(-a,a\right)},\nonumber 
\end{align}
\begin{align}
\left(\sum_{k=1}^{\infty}\left|\frac{P_{\Delta}\nu}{2a\left(\eta+\sigma_{k}\right)\left[\mu\left(\eta+\sigma_{k}\right)+\nu\right]}E_{\alpha}\left(-\left(\mu+\frac{\nu}{\eta+\sigma_{k}}\right)t^{\alpha}\right)l_{k}\right|^{2}\right)^{1/2}\label{eq:dp_tmp_estim2}\\
\leq\frac{P_{\Delta}\nu}{2a\eta\left(\mu\eta+\nu\right)}\left|E_{\alpha}\left(-\left(\mu+\frac{\nu}{\eta+\sigma_{1}}\right)t^{\alpha}\right)\right|\left\Vert K_{1}\right\Vert _{L^{2}\left(-a,a\right)},\nonumber 
\end{align}
with 
\[
K_{1}\left(x\right):=\int_{-a}^{a}K\left(\zeta-x\right)\dd\zeta=2a\left(1+C_{K}\right)-\left(a+x\right)\log\left(a+x\right)-\left(a-x\right)\log\left(a-x\right).
\]
Here, we employed the Parseval's identity $\sum_{k=1}^{\infty}\left|d_{k}^{0}\right|^{2}=\left\Vert p\left(\cdot,0\right)\right\Vert _{L^{2}\left(-a,a\right)}^{2}$,
$\sum_{k=1}^{\infty}\left|l_{k}\right|^{2}=\left\Vert K_{1}\right\Vert _{L^{2}\left(-a,a\right)}^{2}$
and the fact that $E_{\alpha}\left(-t\right)$ is a monotonous function
for sufficiently large values of $t$ (as follows from (\ref{eq:cal_E_a_large})).

Then, thanks to the triangle inequality for the Euclidean $l^{2}$
norm, we estimate (\ref{eq:dp_Parseval}) using (\ref{eq:r_k_def})
and (\ref{eq:dp_tmp_estim1})--(\ref{eq:dp_tmp_estim2}) as
\begin{align}
\left\Vert \delta_{p}\left(\cdot,t\right)\right\Vert _{L^{2}\left(-a,a\right)}\leq & \frac{\nu}{\mu\eta+\nu}\left[\left\Vert p\left(\cdot,0\right)\right\Vert _{L^{2}\left(-a,a\right)}+\frac{P_{\Delta}}{2a\eta}\left\Vert K_{1}\right\Vert _{L^{2}\left(-a,a\right)}\right]\left|E_{\alpha}\left(-\left(\mu+\frac{\nu}{\eta+\sigma_{1}}\right)t^{\alpha}\right)\right|\label{eq:dp_prelim_estim}\\
 & +\frac{P_{\Delta}\nu}{2a\eta^{2}}\left\Vert K_{1}\right\Vert _{L^{2}\left(-a,a\right)}\sup_{k\geq1}\left(\mu+\frac{\nu}{\eta+\sigma_{k}}\right)^{1/\alpha-1}\left|\int_{t}^{\infty}\mathcal{E}_{\alpha}\left(\left(\mu+\frac{\nu}{\eta+\sigma_{k}}\right)^{1/\alpha}\tau\right)\cos\left(\omega\left(t-\tau\right)\right)\dd\tau\right|.\nonumber 
\end{align}

In case $\alpha=1$, the functions $E_{\alpha}$ and $\mathcal{E}_{\alpha}$
in (\ref{eq:dp_prelim_estim}) reduce to exponential functions. Employing
the simple result 
\[
\int_{t}^{\infty}e^{-\beta_{k}\tau}\cos\left(\omega\left(t-\tau\right)\right)\dd\tau=\frac{\beta_{k}e^{-\beta_{k}t}}{\beta_{k}^{2}+\omega^{2}}
\]
with $\beta_{k}:=\left(\mu+\frac{\nu}{\eta+\sigma_{k}}\right)^{1/\alpha}$,
we have
\begin{align*}
\sup_{k\geq1}\left(\mu+\frac{\nu}{\eta+\sigma_{k}}\right)^{1/\alpha-1}\left|\int_{t}^{\infty}\mathcal{E}_{\alpha}\left(\left(\mu+\frac{\nu}{\eta+\sigma_{k}}\right)^{1/\alpha}\tau\right)\cos\left(\omega\left(t-\tau\right)\right)\dd\tau\right| & \leq\sup_{k\geq1}\frac{\beta_{k}^{2-\alpha}e^{-\beta_{k}t}}{\beta_{k}^{2}+\omega^{2}}\\
 & \leq e^{-\beta_{1}t},
\end{align*}
and thus deduce (\ref{eq:dp_conv_alph_1}). 

In case $\alpha\in\left(0,1\right)\cup\left(1,2\right)$, we first
use (\ref{eq:E_a_large}) in the first line of (\ref{eq:dp_prelim_estim})
to deduce the $\mathcal{O}\left(1/t^{\alpha}\right)$ decay of the
corresponding term. Next, we estimate 
\[
\left(\mu+\frac{\nu}{\eta+\sigma_{k}}\right)^{1/\alpha-1}\left|\int_{t}^{\infty}\mathcal{E}_{\alpha}\left(\left(\mu+\frac{\nu}{\eta+\sigma_{k}}\right)^{1/\alpha}\tau\right)\cos\left(\omega\left(t-\tau\right)\right)\dd\tau\right|\leq\frac{C_{\alpha}}{\alpha}\frac{1}{t^{\alpha}},
\]
which is due to the finiteness of the constant $C_{\alpha}:=\sup_{\tau>0}\tau^{1+\alpha}\left|\mathcal{E}_{\alpha}\left(\tau\right)\right|$
entailed by the continuity of $\mathcal{E}_{\alpha}$ (away from $t=0$)
and asymptotics (\ref{eq:cal_E_a_small})--(\ref{eq:cal_E_a_large}).
Therefore, the decay result (\ref{eq:dp_conv}) follows.
\end{proof}

\section{Numerical illustrations\label{sec:numerics}}

We fix the following set of parameters $a=1$, $\nu=2$, $\eta=1$,
$\mu=1.2$, $C_{K}=\log5\simeq1.61$. We take the oscillatory load
profile (\ref{eq:P_t}) with $P_{0}=6$, $P_{\Delta}=0.5$, $\omega=1.5$.
For simplicity, we assume that the punch profile $\Delta\left(x\right)$
is such that 
\begin{equation}
p\left(x,0\right)=\frac{2P_{0}}{a^{2}\pi}\sqrt{a^{2}-x^{2}},\hspace{1em}x\in\left(-a,a\right),\label{eq:p0_sqrt}
\end{equation}
which is a reasonable initial pressure form.

All computations are performed using $60$ terms in the expansion
(\ref{eq:p_sol}).

We illustrate the results for 3 different values of the parameter
$\alpha$ ($\alpha\in\left\{ 0.8,1.0,1.8\right\} $) with the parameter
$\mu=1.2$ and again with $\mu=0$. In the latter case, the model
is purely of a fractional order (no relaxation effect). Also, recall
that when $\alpha=1$, the model reduces to that which does not involve
fractional calculus (with or without relaxation, depending on $\mu$). 

In Figure \ref{fig:stat_profile}, we plot the stationary state pressure
profile (or, more precisely, a collection of curves $p_{\infty}\left(\cdot,t\right)$
evaluated for $t\in\left[0,2\pi/\omega\right]$) and investigate the
dependence of its envelope (given by $\widetilde{p}_{\infty}\pm W_{0}$)
on the choice of model parameters $\alpha$ and $\mu$. Evidently,
for $\mu=1.2$ the impact of the parameter $\alpha$ on the stationary
state is almost undetectable, which is not the case when $\mu=0$. 

Figure \ref{fig:dp_diff_alph} shows the character and the speed of
convergence of the solution to the stationary state $p_{\infty}$
measured by the quantity $\left\Vert \delta_{p}\left(\cdot,t\right)\right\Vert _{L^{\infty}\left(-a,a\right)}:=\left\Vert p\left(\cdot,t\right)-p_{\infty}\left(\cdot,t\right)\right\Vert _{L^{\infty}\left(-a,a\right)}$.
Since the initial profile (\ref{eq:p0_sqrt}) is bounded, the pressure
values at larger times are expected to remain bounded too. This is
why we replaced the $L^{2}$ norm with the $L^{\infty}$ norm in visualising
the solution convergence. Clearly, the results are in direct correspondence
to the analytical prediction given by (\ref{eq:dp_conv_alph_1})--(\ref{eq:dp_conv}).

\begin{figure}
\begin{centering}
\includegraphics[scale=0.62]{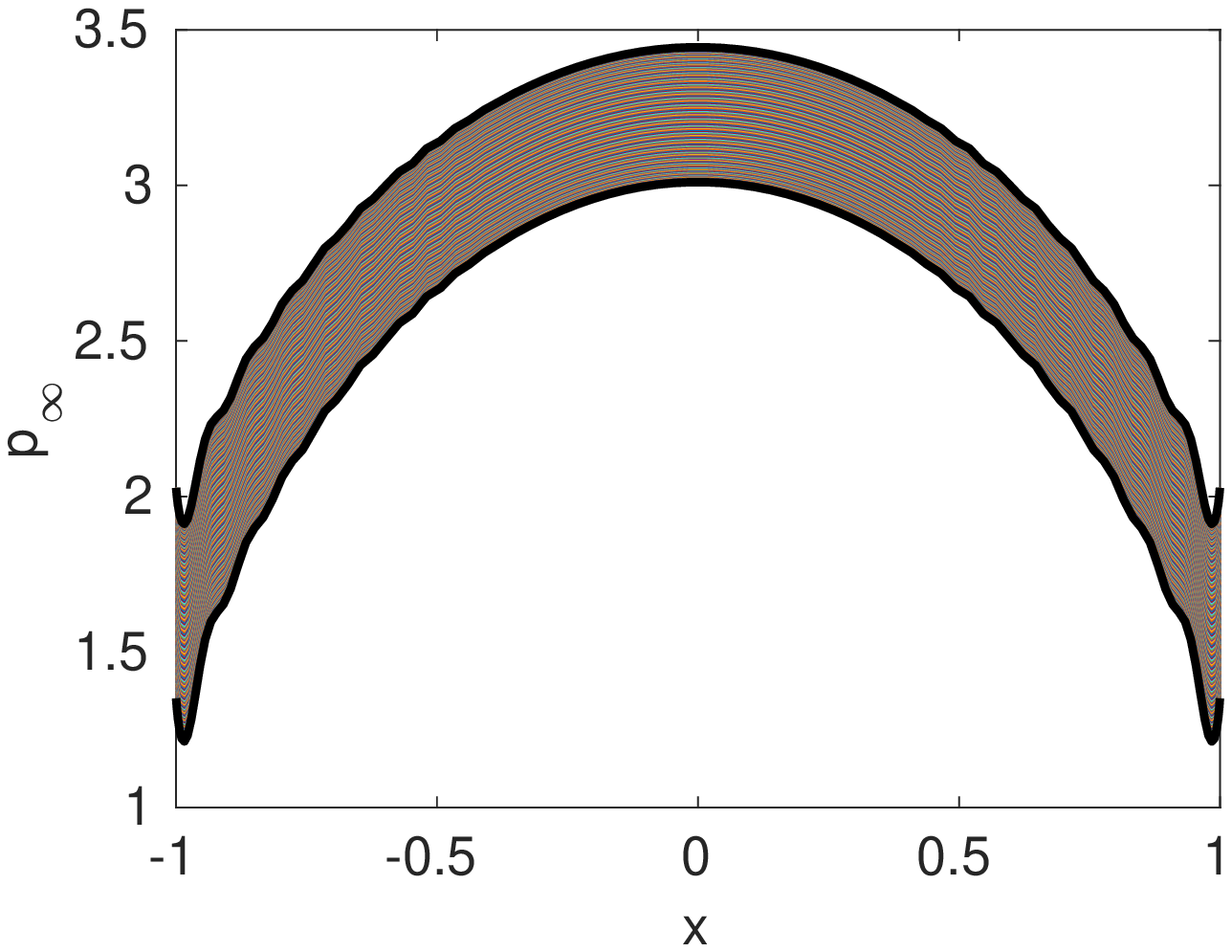}
\par\end{centering}
\centering{}\includegraphics[scale=0.62]{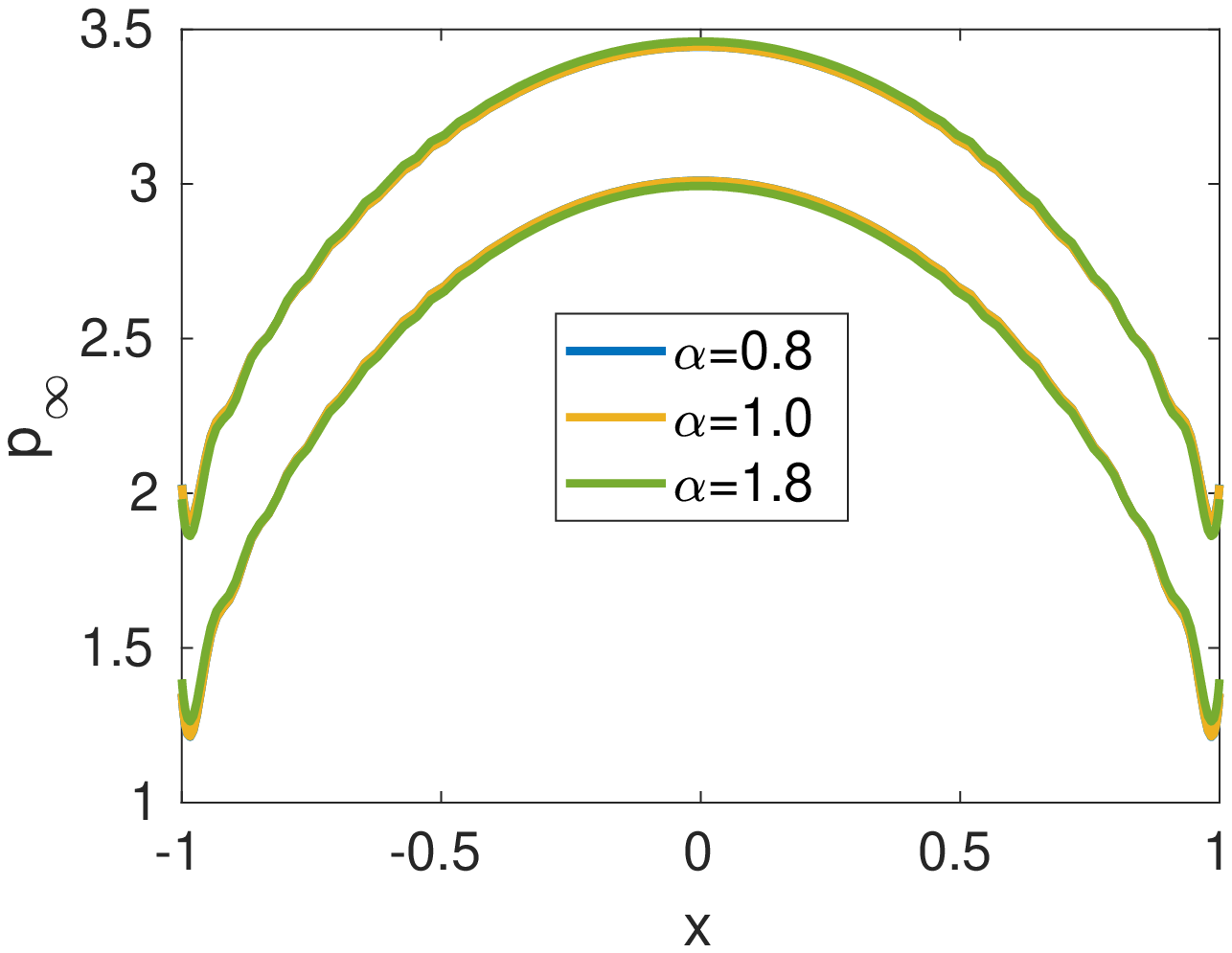}\includegraphics[scale=0.62]{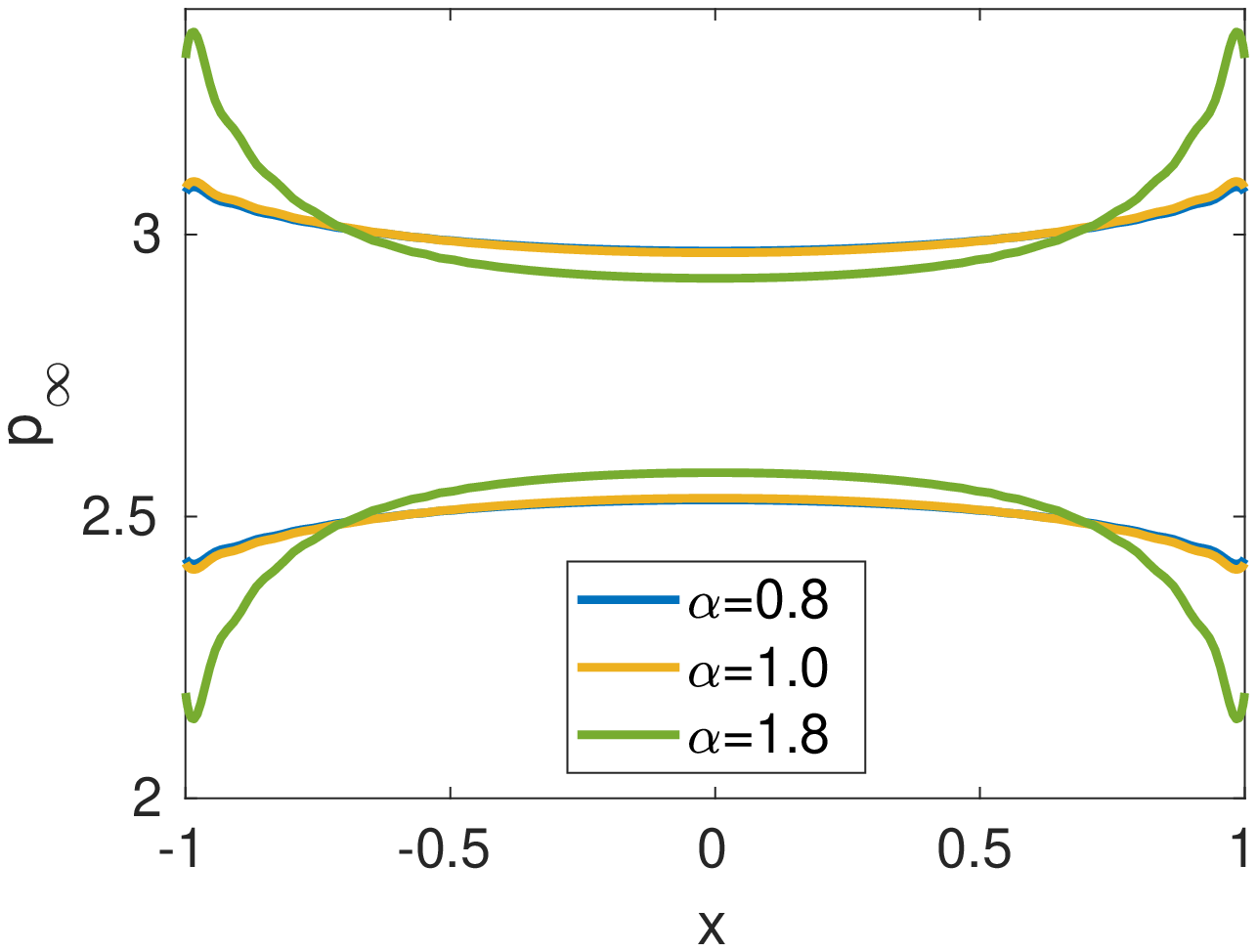}\caption{\label{fig:stat_profile} Stationary state $p_{\infty}\left(x,t\right)$
for multiple values of $t$ over the period $T=2\pi/\omega\simeq4.19$
with $\mu=1.2$ (above) and its envelope for 3 different values of
$\alpha$ (below) with $\mu=1.2$ (left) and $\mu=0$ (right)}
\end{figure}

\begin{figure}
\centering{}\includegraphics[scale=0.62]{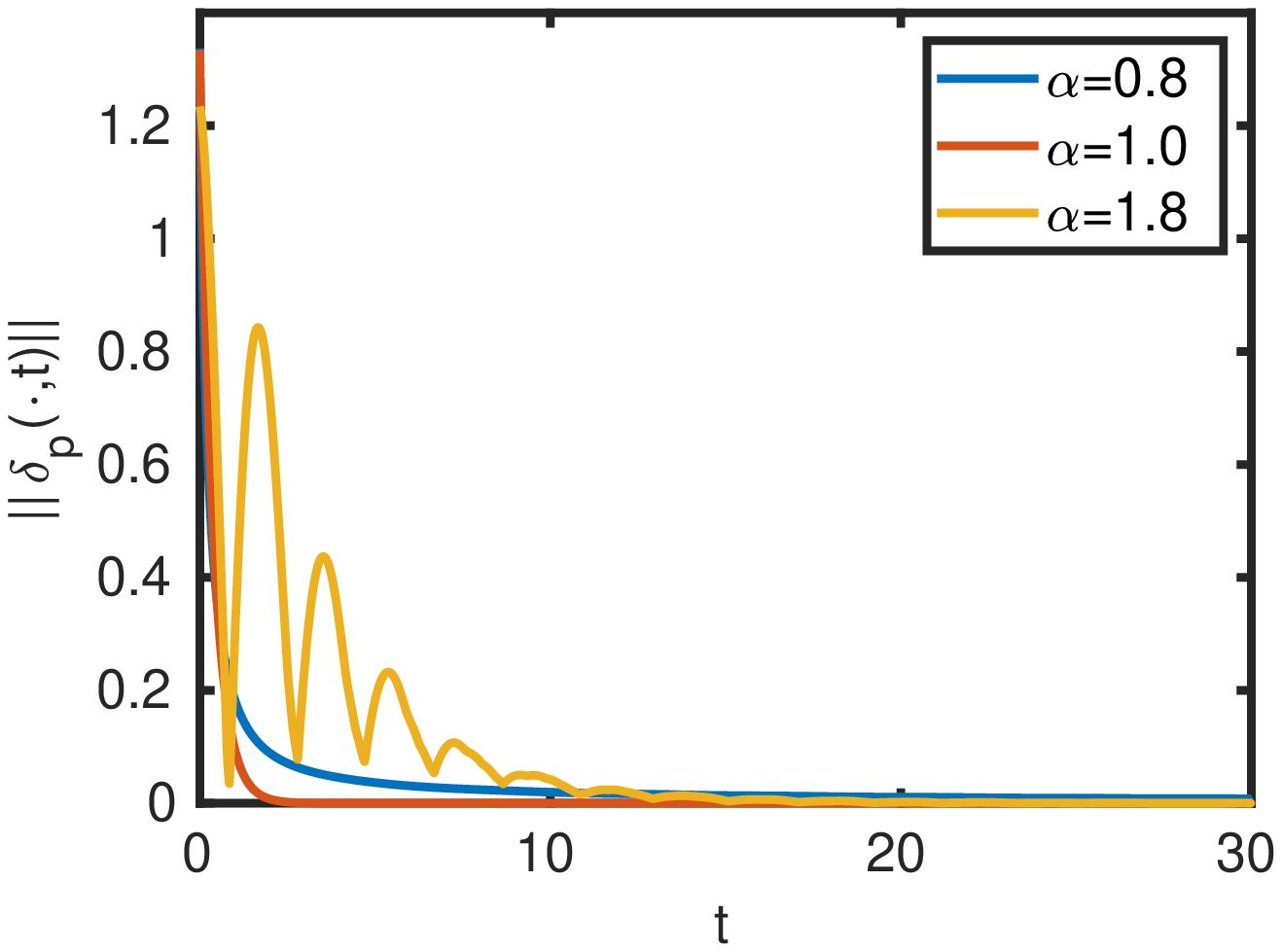}\includegraphics[scale=0.62]{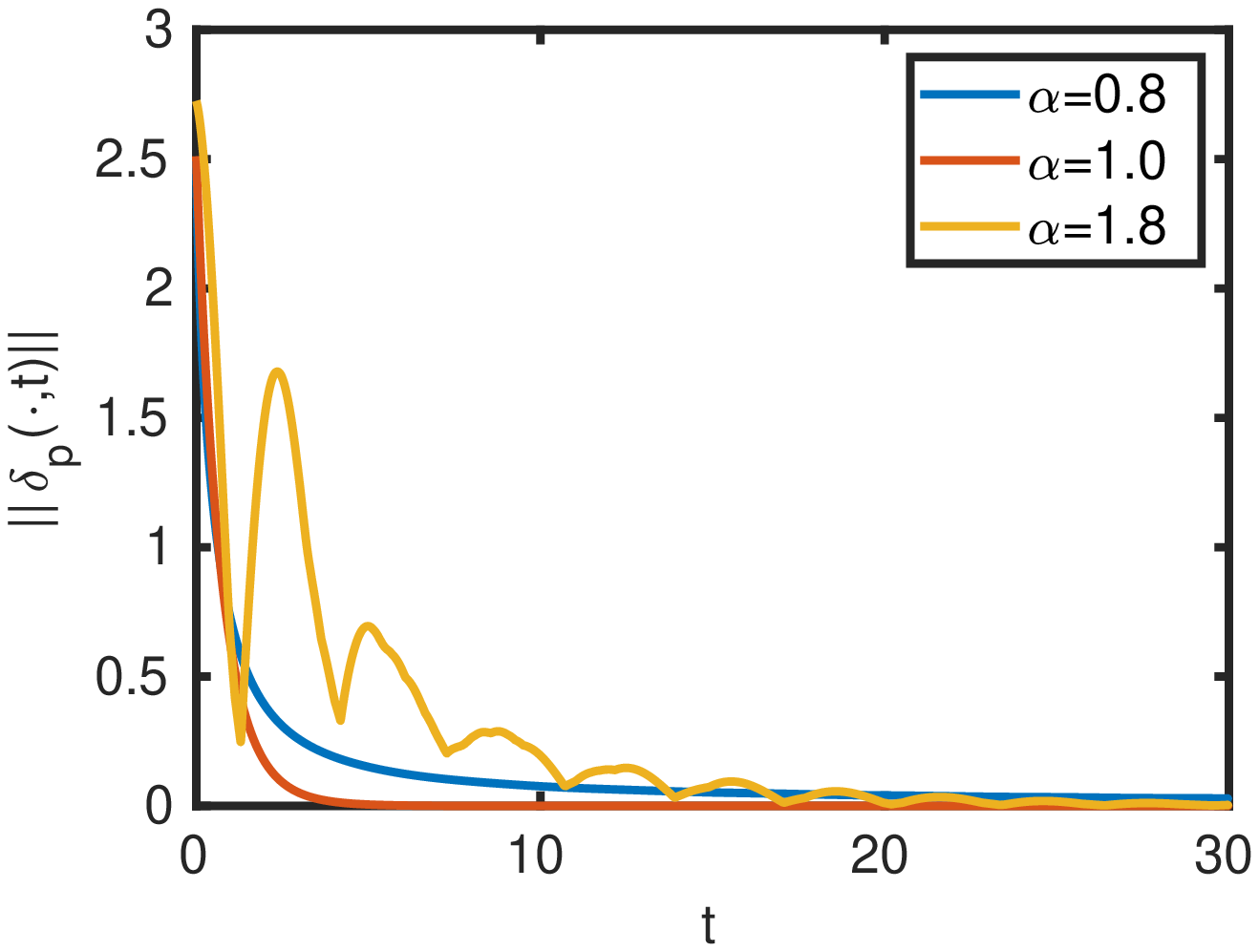}\caption{\label{fig:dp_diff_alph} Convergence of the solution to the stationary
state for 3 different values of $\alpha$ \protect \\
with $\mu=1.2$ (left) and $\mu=0$ (right)}
\end{figure}

\section{Discussion and conclusion\label{sec:conclus}}

We have revisited the classical sliding punch problem with a recently
proposed generalised model of wear. In particular, we have investigated
long-time evolution of the pressure profile under a practically important
case of exterior time-harmonic load. We have derived an explicit form
of the stationary pressure distribution in terms of eigenfunctions
of an auxiliary integral operator. Moreover, we have analysed a speed
of the convergence of the model solution to this distribution. We
note that, in contrast to previous results when the load was constant
(or eventually constant, see \cite[Sec. 5]{Pon}), here the stationary
distribution $p_{\infty}$ is a function of both space and time. Its
time dependence is, nevertheless, clear and structurally simple: it
is harmonic with the same frequency as the exterior load but with
a phase shift that depends on the spatial variable (see (\ref{eq:p_long_t})).

Numerical simulations have been performed to illustrate the obtained
results. In particular, the focus was on the dependence of the mentioned
results on the parameters $\alpha$, $\mu$ which are characteristic
for the present model. It is remarkable that the dependence of the
stationary state on the model order $\alpha$ is insignificant when
$\mu\neq0$. The parameter $\alpha$, however, has an essential impact
on the speed of the convergence towards the stationary state: the
convergence is exponential for $\alpha=1$, whereas for $\alpha\in\left(0,1\right)\cup\left(1,2\right)$,
it is algebraic but its rate grows with the increase of $\alpha$.
Moreover, when $\alpha\in\left(1,2\right)$ the convergence happens
in a non-monotone fashion. The similar effect of $\alpha$ on the
convergence rate was observed in the previous work \cite{Pon} when
the load was constant or eventually constant. We thus confirm here
the previous observation that the model parameters $\mu$ and $\alpha$
affect essentially the stationary state profile and the speed of convergence,
respectively. These statements would constitute important guidelines
when trying to fit the new model to experimental data. Such a fit
would be essential for a practical validation of the model. 

\section*{Acknowledgement}

The author is grateful for the support of AMS {\"O}sterreich during
the period of working on this manuscript.

\end{document}